\pdfoutput=1
\documentclass[11pt]{article}
\def\withcolors{0}
\def\withnotes{0}
\usepackage[T1]{fontenc}
\usepackage[utf8]{inputenc}

\usepackage{lmodern}
\usepackage{xspace}                                     % Smart spacing with \xspace
\usepackage[protrusion=true,expansion=true]{microtype}  % Improve font rendering

\usepackage[normalem]{ulem}

\usepackage{amsfonts,amsmath,amssymb, amsthm, mathtools}
\usepackage{thm-restate}
\usepackage{dsfont} % For the indicator symbol

\usepackage{algorithmicx,algpseudocode,algorithm}

\usepackage[usenames,dvipsnames,table]{xcolor}

\usepackage{csquotes}

\usepackage{relsize}

\usepackage{multirow}
\usepackage{chngpage} % allows for temporary adjustment of side margins

\usepackage{mfirstuc}

\usepackage{tikz}
\usetikzlibrary{arrows}
\usetikzlibrary{calc,decorations.pathmorphing,patterns}

\usepackage[backref,colorlinks,citecolor=blue,bookmarks=true,linktocpage]{hyperref}
\usepackage{aliascnt}
\usepackage[numbered]{bookmark}

\usepackage{accents}

\usepackage{fullpage}

\usepackage{titling}

\usepackage[shortlabels]{enumitem}
  \setitemize{noitemsep,topsep=3pt,parsep=2pt,partopsep=2pt} % Uncomment for compact item lists
  \setenumerate{itemsep=1pt,topsep=2pt,parsep=2pt,partopsep=2pt}
  \setdescription{itemsep=1pt}
  
\ifnum\withnotes=1
  \usepackage[colorinlistoftodos,textsize=scriptsize]{todonotes}
\fi

\usepackage{verbatim}

\usepackage{mleftright} % \mleft( #1 \mright)
\makeatletter
\@ifundefined{theorem}{%
  \theoremstyle{plain} %% Style
  	\newtheorem{theorem}{Theorem}[section]
  	\newaliascnt{coro}{theorem}
  	  \newtheorem{corollary}[coro]{Corollary}
  	\aliascntresetthe{coro}
  	\newaliascnt{lem}{theorem}
  		
  	\aliascntresetthe{lem}
  	\newaliascnt{clm}{theorem}
  		\newtheorem{claim}[clm]{Claim}
	\aliascntresetthe{clm}
	\newaliascnt{fact}{theorem}
 	 	\newtheorem{fact}[theorem]{Fact}
	\aliascntresetthe{fact}
  	
  		\newaliascnt{ques}{theorem}
  	\newtheorem{question}[ques]{Question}
  \newaliascnt{prop}{theorem}
  	
	\aliascntresetthe{prop}
	\newaliascnt{conj}{theorem}
  		
	\aliascntresetthe{conj}
  \theoremstyle{remark} %% Style

  \theoremstyle{definition} %% Style
  	\newaliascnt{defn}{theorem}
 		 \newtheorem{definition}[defn]{Definition}
 	 \aliascntresetthe{defn}
}{}
\makeatother
\newenvironment{proofof}[1]{\begin{proof}[Proof of {#1}]}{\end{proof}}

\providecommand{\email}[1]{\href{mailto:#1}{\nolinkurl{#1}\xspace}}

\ifnum\withcolors=1
  \newcommand{\acolor}[1]{{\color{orange}#1}} % Akash
  \newcommand{\ccolor}[1]{{\color{black!40!cyan}#1}} % Clement
  \newcommand{\ecolor}[1]{{\color{black!50!green}#1}} % Elena
  \newcommand{\kcolor}[1]{{\color{Plum}#1}} % Karl
  \newcommand{\scolor}[1]{{\color{RubineRed}#1}} % Siyao
\else

  \newcommand{\acolor}[1]{{#1}}
  \newcommand{\ccolor}[1]{{#1}}
  \newcommand{\ecolor}[1]{{#1}}
  \newcommand{\kcolor}[1]{{#1}}
  \newcommand{\scolor}[1]{{#1}}
\fi

\ifnum\withnotes=1
  \newcommand{\anote}[1]{\par\acolor{\textbf{A: }\sf #1}} % Akash
  \newcommand{\cnote}[1]{\par\ccolor{\textbf{C: }\sf #1}} % Clement
  \newcommand{\enote}[1]{\par\ecolor{\textbf{E: }\sf #1}} % Elena
  \newcommand{\knote}[1]{\par\kcolor{\textbf{K: }\sf #1}} % Karl
  \newcommand{\snote}[1]{\par\scolor{\textbf{S: }\sf #1}} % Siyao

\else
  \newcommand{\anote}[1]{}
  \newcommand{\cnote}[1]{}
  \newcommand{\enote}[1]{}
  \newcommand{\knote}[1]{}
  \newcommand{\snote}[1]{}

\fi
\newcommand{\ignore}[1]{\leavevmode\unskip} % eat unnecessary spaces before
\newcommand{\eps}{\ensuremath{\varepsilon}\xspace}
\newcommand{\Tester}{\ensuremath{\mathcal{T}}\xspace} % Testing algorithm T
\newcommand{\property}{\ensuremath{\mathcal{P}}\xspace} % Property P
\newcommand{\eqdef}{\stackrel{\rm def}{=}}

\newcommand{\littleO}[1]{{o\mleft( #1 \mright)}}
\newcommand{\bigO}[1]{{O\mleft( #1 \mright)}}

\newcommand{\bigOmega}[1]{{\Omega\mleft( #1 \mright)}}

\newcommand{\tildeO}[1]{\tilde{O}\mleft( #1 \mright)}

\providecommand{\poly}{\operatorname*{poly}}

\newcommand{\setOfSuchThat}[2]{ \left\{\; #1 \;\colon\; #2\; \right\} } 			% sets such as "{ elems | condition }"
\newcommand{\dist}[2]{\operatorname{dist}\!\left(#1, #2\right)}

\newcommand\restr[2]{{% we make the whole thing an ordinary symbol
  \left.\kern-\nulldelimiterspace % automatically resize the bar with \right
  #1 % the function
  \vphantom{\big|} % pretend it's a little taller at normal size
  \right|_{#2} % this is the delimiter
  }}

\newcommand{\proba}{\Pr}

\newcommand{\probaDistrOf}[2]{\proba_{#1}\left[\, #2\, \right]}

\newcommand{\expect}[1]{\mathbb{E}\!\left[#1\right]}

\newcommand{\abs}[1]{\left\lvert #1 \right\rvert}

\newcommand{\R}{\ensuremath{\mathbb{R}}\xspace}

\newcommand{\pdfsamp}{dual\xspace}
\newcommand{\cdfsamp}{cumulative dual\xspace}
\newcommand{\Pdfsamp}{\expandafter\capitalisewords\expandafter{\pdfsamp}}
\newcommand{\Cdfsamp}{\expandafter\capitalisewords\expandafter{\cdfsamp}}

\makeatletter

\newcommand{\Rom}[1]{\expandafter\@slowromancap\romannumeral #1@}

\makeatother

\newcommand{\bX}{\mathbf{X}}
\newcommand{\bx}{\mathbf{x}}

\newcommand{\by}{\mathbf{y}}
\newcommand{\bZ}{\mathbf{Z}}
\newcommand{\bz}{\mathbf{z}}

\newcommand{\bh}{\mathbf{h}}

\newcommand{\bS}{\mathbf{S}}

\def\cl#1{{\cal #1}} % caligraphic

\makeatletter
\newcommand\ackname{Acknowledgements}
\if@titlepage
  \newenvironment{acknowledgements}{%
      \titlepage
      \null\vfil
      \@beginparpenalty\@lowpenalty
      \begin{center}%
        \bfseries \ackname
        \@endparpenalty\@M
      \end{center}}%
     {\par\vfil\null\endtitlepage}
\else
  
\fi
\makeatother
\def\authornameeg{Elena Grigorescu}
\def\authoraffieg{Purdue University. Email: \email{elena-g@purdue.edu}. Research supported in part by NSF CCF-1649515.}
\def\authornameak{Akash Kumar}
\def\authoraffiak{Purdue University. Email: \email{akumar@purdue.edu}. Research supported in part by NSF CCF-1649515 and NSF CCF-1618918.}
\def\authornamekw{Karl Wimmer}
\def\authoraffikw{Duquesne University. Email: \email{wimmerk@duq.edu}.}

\title{Flipping out with many flips: hardness of testing $k$-monotonicity}
\date{\today}
 
\author{
  \ecolor{\authornameeg}\thanks{\authoraffieg}
  \and \acolor{\authornameak}\thanks{\authoraffiak}
  \and \kcolor{\authornamekw}\thanks{\authoraffikw}
}

\begin{document}

\maketitle

\begin{abstract}

A function $f:\{0,1\}^n\rightarrow \{0,1\}$  is said to be $k$-monotone if it flips between $0$ and $1$ at most $k$ times on every ascending chain. Such functions represent a natural generalization of ($1$-)monotone functions, and have been recently studied in circuit complexity, PAC learning, and cryptography. Our work is part of a renewed focus in understanding testability of properties characterized by freeness of arbitrary order patterns as a generalization of monotonicity. Recently,  Canonne et al. (ITCS 2017) initiate the study of $k$-monotone functions in the area of property testing, and  Newman et al. (SODA 2017) study testability of families characterized by freeness from order patterns on real-valued functions over the line $[n]$ domain.  

We study $k$-monotone functions in the more relaxed {\em parametrized  property testing model}, introduced by Parnas et al. (JCSS, 72(6), 2006). In this process we resolve a problem left open in previous work.  Specifically, our results include the following.

\begin{enumerate}
\item Testing $2$-monotonicity  on the hypercube  non-adaptively with one-sided error requires an exponential in $\sqrt{n}$ number of queries. This behavior shows a stark contrast with testing (1-)monotonicity, which only needs $\tildeO{\sqrt{n}}$ queries (Khot et al. (FOCS 2015)).
Furthermore, even the apparently easier task of distinguishing $2$-monotone functions from functions that are far from being $n^{.01}$-monotone also requires an exponential number of queries.
\item On the hypercube $[n]^d$ domain, there exists a testing algorithm that makes a constant number of queries and distinguishes functions  that are $k$-monotone from functions that are far from being $O(kd^2) $-monotone.  Such a dependency is likely necessary, given the lower bound above for the hypercube.

\end{enumerate}

\end{abstract}

\ifnum\withnotes=1
 % \clearpage
  %\listoftodos
  %\hfill
  %\todonoteinline{Check capitalization and consistency in the bibliography.}
  %\clearpage
\fi
%\clearpage
%\tableofcontents
\clearpage

\section{Introduction}

{\em Property testing} \cite{BlumLR93, RubinfeldS:96, GoldreichGR98} studies the complexity of deciding if a large object satisfies a property, or is far from satisfying the property, when the algorithm has only partial access to its input.
Two prolific lines of research in the study of Boolean functions have recently seen ultimate results: on one hand, the study of families exhibiting algebraic symmetries (such as low-degree polynomials, and triangle-freeness), explicitly initiated in \cite{KS08}, and on the other hand, the study of functions with less symmetry, but whose values respect a monotone order relation, initiated in \cite{GGLRS00}.
Extending the structural features of these classes of functions,  in this work we view families of Boolean functions with less symmetry as interpolating from basic families of monotone functions to families characterized by freeness of  more complex {\em order patterns}. This perspective is apparent in the recent works of   Newman et al \cite{NewmanRRS:17} and Canonne et al \cite{CGGKW:17} (a superset of the authors) who propose the study of families exhibiting freeness from more general order patterns in the property testing setting. Here we study the property of being {\em $k$-monotone}, building towards  understanding testability of freeness from arbitrary order patterns.

We first introduce some standard definitions. A \emph{property} of Boolean functions from a discrete domain $D$ to $\{0,1\}$ is a subset of  $\{f\colon D \rightarrow \{0,1\}\}$. Given two functions $f,g\colon D \to \{0,1\}$,  denote by $d(f,g)$  the (normalized) Hamming distance between them, i.e. $d(f,g) = \probaDistrOf{x\sim D}{ f(x) \neq g(x) }.$   The distance of a function $f$ to $\property$ is defined as $d(f, \property) = \min_{g\in\property} d(f,g).$  A \emph{$q$-query tester} for $\property$ is a randomized algorithm $\Tester$ that takes as input $\eps\in(0,1]$, and has query access to a function $f\colon D\to \{0,1\}$.  After making at most $q(\eps)$ queries, \Tester  distinguishes between the following two cases: i)  if $f \in \property$, then \Tester accepts, with probability $2/3$; ii)  if $d(f, \property) \geq \eps$, then \Tester  rejects  w.p. $2/3$. If the algorithm only errs in the second case but accepts any function $f\in\property$ with probability $1$, it is said to be \emph{one-sided}; otherwise, it is said to be \emph{two-sided}. Moreover, if the queries made to the function can only depend on the internal randomness of the algorithm, but not on the values obtained during previous queries, it is said to be \emph{non-adaptive}; otherwise, it is \emph{adaptive}. The maximum number of queries made to $f$ in the worst case is the \emph{query complexity} of the testing algorithm. When $d(f, \property) \geq \epsilon$, we say that $f$ is \emph{$\epsilon$-far} from $\property$. We will frequently use ``far'' to denote ``$\Omega(1)$-far''.

In this work we focus on the property of being $k$-monotone, which we formalize next.

 \begin{definition}[\bf{$k$-monotonicity}]
		 A function $f: D \rightarrow \{0,1\}$ is {\em $k$-monotone} if there do not exist $x_1\prec x_2\prec \ldots \prec x_{k+1}$ in $D$, such that $f(x_1)=1$ and $f(x_i)\ne f(x_{i+1})$ for all $i\in [k]$. Equivalently, a Boolean $k$-monotone function is free of the pattern $\langle f(x_1), f(x_2), \ldots, f(x_{k+1})\rangle=\langle 1, 0, 1, 0, \ldots, 0, 1\rangle$ \footnote{Note that $f(x_{k+1})$ could be $0$ depending on the parity of $k$}.
 Note that $1$-monotone functions are exactly the functions that are monotone.%, and hence free of the pattern $\langle f(x_1), f(x_2)\rangle =\langle 1, 0\rangle.$ 

 \end{definition}
 
 Testing monotonicity has drawn interest in the theoretical computer science community for almost two decades (e.g., \cite{GGLRS00, DGLRRS99, ErgunKKRV00, FLNRRS:02,  Fischer04, BatuRW05, AilonC06, HalevyK08, BhattacharyyaGJRW:09, BrietCGM10, FattalR:10,  BlaisBM12, CS:13a, CS:13b, CS:13, BermanRY:14, ChenST:14, ChenDST:15, KhotMS:15, BelovsB:15, ChenWX:17}), especially due to the naturalness of the property, as well as its evasiveness to tight analysis. 

  The notion of $k$-monotonicity has been studied ever since the 50's in the context of circuit lower bounds \cite{Markov:57}.
  Indeed, $k$-monotone functions correspond to functions computable by Boolean circuits with $\log k$ negation gates, and in particular monotone functions correspond to circuits with no negation gates. In proving lower bounds for circuits with few negation gates, it has been apparent that the presence of even one negation gate leads either to failure of the common analysis techniques,  or to failure of the expected results, as remarked by Jukna \cite{Jukna:12}.

Interest in $k$-monotonicity has been recently rekindled from multiple angles, including PAC learning, circuit complexity, cryptography and Fourier analysis \cite{Rossman:15, GuoMOR:15, GuoK:15, LinZ:16, DineshS:18}. 
In these areas, as $k$ increases, $k$-monotone functions are viewed as robust classes sitting between the very structured monotone functions, and general Boolean functions. 
The study of $k$-monotonicity in the property testing model was recently initiated in  \cite{CGGKW:17}. Here we continue this study and reveal a stark difference from testing monotonicity, while disproving a conjecture from \cite{CGGKW:17}. Our study leads to open questions that might be relevant to monotonicity testing.

\subsection{Testing $2$-monotonicity on the hypercube} \label{sec:2-mon-on-cube}

Recall that a function $f:\{0,1\}^n\rightarrow \{0,1\}$ is $2$-monotone if it is free from  $\langle f(x_1), f(x_2), f(x_3) \rangle= \langle 1, 0, 1\rangle$, for $x_1\prec x_2 \prec x_3$. How much more difficult can one-sided testing of this pattern be than testing freeness from $\langle f(x_1), f(x_2) \rangle= \langle 1, 0\rangle$ (i.e., monotonicity)?

The first proof of the $O(n)$-query one-sided non-adaptive tester for monotonicity from~\cite{GGLRS00} relies on the following structural theorem: if $f$ is far from monotone, then there are many \emph{edges} that contain a  violation to monotonicity (i.e., $x_1\prec x_2$, with $\langle f(x_1), f(x_2) \rangle= \langle 1, 0\rangle$, and there is no $x_3$ such that $x_1 \prec x_3 \prec x_2$).  For $k$-monotonicity with $k \geq 2$, there is no such result: since all violations require at least $3$ points, the violations can all be spread across many levels of the cube.  For example, consider a totally symmetric  \emph{block function} $f(x)$ such that $f(x) = 1$ if $|x|$ is between $n/2 - 2\sqrt{n}$ and $n/2 - \sqrt{n}$, or between $n/2 + \sqrt{n}$ and $n/2 + 2\sqrt{n}$, and $f(x) = 0$ otherwise.  This function is far from $2$-monotone, yet any triple of points that witnesses this fact will contain a pair of points whose Hamming distance is at least $\sqrt{n}$. 

This non-locality of a violation seems to be a reason why  testing $k$-monotonicity (non-adaptively, with one-sided error) turns out to be very different than just testing monotonicity.
 Initial lower bounds from \cite{CGGKW:17} only show a separation for $k\geq 3$. Indeed,  
 they ~\cite{CGGKW:17} show that testing $k$-monotonicity non-adaptively with one sided error requires $\Omega(n/k^2)^{k/4}$ queries, which for $k\geq 3$  beats the recent tester for monotonicity of $\tildeO{\sqrt n}$ query complexity \cite{CS:13b, KhotMS:15}. 
 However, no such separation was known for $k=2$. 
 Furthermore, they conjecture that $k$-monotonicity gets gradually more difficult to test as $k$ increases, and that $k$-monotonicity should be testable with one-sided error with a budget of $\Theta(n^{O(k)})$ queries. In this work we show that this is not the case, and in fact even testing $2$-monotonicity itself requires an exponential number of queries.

\begin{theorem} (Rephrasing Corollary \ref{cor:lb:na:1s:new})
Testing $2$-monotonicity non-adaptively with one-sided error requires $2^{\Omega(\sqrt{n})}$ queries. 
\end{theorem}

For comparison,  \cite{CGGKW:17}  obtains a tester for $2$-monotonicity  with  $2^{O(\sqrt{n}\cdot \log n)}$ queries, and  \cite{BlaisCOST:15} obtains a $2^{\Omega( \sqrt{n})}$ lower bound for PAC learning. 
Hence, testing $2$-monotonicity is a natural property for which testing is essentially as hard as PAC learning.

While being contrary to the intuition that since monotonicity testing is easy, so should $2$-monotonicity be, this result  reinforces the theme discussed above in  proving circuit complexity lower bounds.  The class of $2$-monotone functions is exactly the class of functions computable by a Boolean circuit with at most one negation gate~\cite{Markov:57,BlaisCOST:15}.  As in the circuit complexity world, allowing one negation gate significantly increases complexity.\\

\noindent \textbf{A canonical test for one-sided testing}\\

\noindent As alluded to above, in order to test $k$-monotonicity with one-sided error, a tester should only reject if it finds a {\em violation} in the form of a sequence $x_1\prec x_2\prec \ldots \prec x_{k+1}$ in $\{0,1\}^n$, such that $f(x_1)=1$ and $f(x_i)\ne f(x_{i+1})$ for all $i\in [k]$. Hence, a canonical candidate tester suggested in \cite{CGGKW:17}  queries all points along a random chain  and rejects only if it finds a violation. 

\begin{definition}		We define the (basic) \emph{chain tester} to be the algorithm that picks a \emph{uniformly random chain} $\bZ=\langle 0^n \prec \bz_1 \prec \bz_2 \prec \cdots \bz_{n-1} \prec 1^n \rangle$ of comparable points from $\{0,1\}^n$, and queries $f$ at all these points.  The chain tester rejects if $\bZ$ reveals a violation to $k$-monotonicity, otherwise it accepts. We also sometimes denote by a chain tester an algorithm that picks multiple chains (possibly from a joint distribution).
\end{definition}
	
All previous tests for monotonicity (e.g., \cite{GGLRS00, CS:13a, ChenST:14, KhotMS:15}) imply a chain tester  incurring only a small polynomial blow-up in the query complexity.

As expected, a chain tester is indeed implied by any non-adaptive one-sided algorithm for $k$-monotonicity.

\begin{theorem} \label{thm:canon1}(Corollary to Theorem \ref{thm:lb:na:1s:k-vs-far-from-g(k,n)})	Any non-adaptive one-sided $q$-query tester for $k$-monotonicity over $\{0,1\}^n$ implies an $\bigO{q^{k+1}n}$-query tester that queries points on a distribution over chains, and 
succeeds with constant probability. In particular, if $p$ is the success probability of the basic chain tester, then $p=\Omega(1/q^{k+1})$.
\end{theorem}

To prove our lower bounds, we create families for which the chain tester fails. We first remark that the chain tester does very well on the `block functions'  described above in Section ~\ref{sec:2-mon-on-cube} that are far from being $2$-monotone. Indeed, every chain from $0^n$ to $1^n$ will uncover a violation to $2$-monotonicity!  
In our constructions, we get around this by hiding such functions with ``long'' violations on a {\em  small set of coordinates}, while still making sure it comprises a constant fraction of the cube.  We show that a random chain is unlikely to visit enough of these coordinates to find a violation.

Proving that these functions are far from $k$-monotone amounts to understanding the structure of the violation hypergraph (i.e., the hypergraph  whose vertices are elements of $\{0,1\}^n$, and whose edges are the tuples that witness a violation).  
The violation graph is to some extent the only handle we have on arguing structural properties of functions that are far from being $k$-monotone. A large matching (edges with disjoint sets of vertices) in this hypergraph implies that the function is far from $k$-monotone. Indeed, such families can be shown to have a large matching. 

In fact, our results work in a larger context of parametrized testing, and our lower bounds hold even for apparently much easier tasks, as we describe next.

\subsection{Parametrized monotonicity testing on the hypercube}

The notion of $k$-monotonicity allows us to propose the  natural problem of approximating the ``monotonicity'' of a function.  For a concrete example, suppose we are promised that the unknown function $f$ is either $2$-monotone or far from, say, $n^{0.01}$-monotone.  That is, either $f$ changes value at most twice on every chain, or a constant fraction of points of $\{0,1\}^n$ need to be changed so that $f$ changes value at most $n^{0.01}$ times on every chain.
This promise problem can only require fewer queries, and intuitively, it should be much fewer.

However, we show that intuition is incorrect: even the apparently much easier task of distinguishing between functions that are $2$-monotone and functions that are far from being $n^{.01}$-monotone requires an exponential number of queries.

\begin{theorem} (Rephrasing Corollary \ref{cor:lb:na:1s:2-vs-far-from-k})
Testing non-adaptively with one-sided error whether a function is  $2$-monotone or far from being $n^{0.01}$-monotone requires $2^{\Omega(n^{0.48})}$ queries.
\end{theorem}

This is also the setting of parametrized property testing introduced by  Parnas et al~\cite{PRR:06}.
In this setting, a property is parametrized by an integer  $k$ and denoted $\property=\{\property_k\}_k$. A \emph{$(k_1,k_2)$-tester} for the family $\property$ is a randomized algorithm which, on input a proximity parameter $\eps\in(0,1)$ and oracle access to an unknown function $f$,  must accept if $f\in\property_{k_1}$,  with probability at least $2/3$; and reject  if $d(f, \property_{k_2})>\eps$, with probability at least $2/3$.

Hence, we are interested  in $(k, g(k,n))$-testing for parametrized monotonicity, denoted $\{{\mathcal M}_k\}_k$. Following the notation just described, we will speak of property testers for (and query complexity lower bounds for) $(k, g(k,n))$-testing the property $\{{\mathcal M}_k\}_k$.

\begin{theorem}\label{cor:lb:k-vs-g(k,n):general}% (Rephrasing Corollary ~\ref{cor:lb:k-vs-g(k,n):general})
		Given $2\leq k \leq g(k,n) =\littleO{\sqrt n}$, $(k, g(k,n))$-testing $\{{\mathcal M}_k\}_k$ (on the hypercube) non-adaptively with one-sided error requires $2^{\Omega\left(\frac{\sqrt n}{(k+1)(g(k,n)/k)^2}\right)}$ queries.
\end{theorem}

Parametrized monotonicity testing may provide a new angle for approaching yet unanswered questions towards the goal of understanding the role of adaptivity in testing monotonicity. The current strong lower bounds for testing monotonicity adaptively \cite{ChenST:14, ChenDST:15,  BelovsB:15} come polynomially close to the one-sided error upper bounds \cite{KhotMS:15}, yet the question of whether adaptive algorithms can beat the current lower bounds still remain open.

\subsection{Parametrized monotonicity testing on the hypergrid}

We next study parametrized monotonicity testing of Boolean functions  over the hypergrid $[n]^d$ domain. For these domains \cite{CGGKW:17} shows that testing $k$-monotonicity non-adaptively with two-sided error can be performed with a number of queries independent of $n$, but {\em exponential} in the dimension $d$. More explicitly, $q(n, d, \eps, k)=\min \left( \tildeO{\frac{1}{\eps^2} (5kd/\eps)^d}, 2^{\tildeO{k\sqrt{d}/\eps^2}}\right).$

Our algorithmic results show a contrast to the hypercube case. We obtain a tester with {\em constant} query complexity (independent of $k, n, d$), albeit trading off for the $(k_1, k_2)$ parameters that the tester needs to distinguish between.

\begin{theorem}\label{thm:ub}
	There is a non-adaptive two-sided tester which performs $\poly(1/\eps)$  queries and $(k, 2kd^2/\eps)$-tests $\{{\mathcal M}_k\}_k$ on the hypergrid $[n]^d$.
\end{theorem}

Two-sided versus one-sided testing for monotonicity on the hypercube is an unresolved problem. Thus, in light of our exponential  bounds for the $\{0, 1\}^n$ domain, the dependence in $d$ in the $g(n,d, k)=2kd^2$ function appears necessary.  Achieving a sublinear dependence on $d$ while controlling the query complexity would require new ideas that also apply to the hypercube.

\subsection{Related work and the larger perspective}

As previously mentioned, \cite{NewmanRRS:17} proposes studying more general order patterns for real-valued functions defined over the line domain $[n]$. They view a pattern of length $k$ as a permutation $\pi:[k]\rightarrow [k]$. A function $f:[n]\rightarrow \R$ if $\pi$-free if there do not exist indices $i_1< i_2<\ldots < i_k$ such that $f(i_x) < f(i_y)$ whenever $\pi(x)<\pi(y)$. This is a more fine-grained notion than $k$-monotonicity; for example, $(2,1)$-freeness describes monotonicity, and the intersection of $(2, 1, 3)$-freeness and $(3, 1, 2)$-freeness describes $2$-monotonicity over the reals.  The authors of~\cite{NewmanRRS:17} obtain several results for non-adaptive one-sided testing on the line, essentially distinguishing monotone and non-monotone patterns. They further raise the question of characterizing the testing complexity as a function of the structure of the pattern.  Their lower bounds are $\Omega(\sqrt{n})$; this is quite strong given that the domain size is $n$, and this bound is likely not tight for most patterns. Recently, \cite{Ben-EliezerC:18} showed that for ``most'' permutations $\pi$ testing $\pi$-freeness requires $\frac{n^{1 - 1/(k - \theta(1))}}{\eps^{1/(k - \theta(1) )}}$ queries.                                                                    
Transferring this question to our setting, we may ask:

\begin{question}\label{q:patterns}
On the hypercube, do all order patterns (longer than some constant length) require $\exp(\Omega(\sqrt{n}))$ queries to test?  Can this lower bound be strengthened to $\exp(\Omega(n))$ for most patterns?  Can these patterns be characterized?
\end{question}

The description of a function family as being free from a pattern is common in property testing of Boolean functions, especially in the study of families that exhibit `affine-invariance'. In that line of work, the Boolean domain $\{0,1\}^n$ is viewed as a vector space over the field of two elements, and now one may perform algebraic operations over the domain. For instance, the property of being `triangle-free' is defined as freeness from the pattern $\langle f(x), f(y), f(z)\rangle=\langle 1, 1, 1\rangle $ whenever $x+y+z=0$.
 An almost complete characterization of when such a property is testable with a constant number of queries has recently been obtained as a result of sustained interest in this quest \cite{Green05, KralSV09, Shapira10,  BhattacharyyaCSX11,BhattacharyyaFL13, BhattacharyyaFHHL13, Yoshida14,  BhattacharyyaGS15}.

Monotonicity and $k$-monotonicity however exhibit less symmetric structure, but these families are still invariant under permutations of the variables. Invariance under permutations of the variables is also maintained 
by properties defined by freeness of order patterns. While this symmetry is not so strong as to imply constant query testers in general, it is however crucial in the design of algorithms for such properties. For example, we use this invariance in showing the reduction to a canonical tester in Theorem \ref{thm:canon1}.

The notion of $k$-monotonicity can be easily extended to real-valued functions. In fact, this extension has already implied `tolerant' monotonicity testers for  families consisting of functions $f\colon [n]^d\rightarrow [0,1]$ \cite{CGGKW:17}.
Furthermore, \cite{CGGKW:17} reveals connections between testing $k$-monotonicity and  testing surface area \cite{KearnsR:00,BalcanBBY:12,KothariNOW:14, Neeman:14}, as well as estimating support of distributions \cite{CanonneR:14}.

Besides the relevant extensive literature on monotonicity testing mentioned before, another recent direction that generalizes monotonicity is testing {\em unateness}. Namely, a unate function is monotone on each full chain, however, edge-disjoint chains may be monotone in different directions. Initiated in \cite{GGLRS00}, query-efficient unateness testers were obtained in \cite{GGLRS00, CS:16a, KS:16, BaleshzarCPRS17}, and lower bounds in~\cite{ChenWX:17, BaleshzarCPRS17}.%  shows a matching one-sided non-adaptive $\tildeOmega{n}$ bound.

A trickle of recent work in cryptography focuses on understanding how many negation gates are needed to compute cryptographic primitives, such as one-way permutations, small bias generators, hard core bits, and extractors \cite{GuoMOR:15, GuoK:15, LinZ:16}. Very recently,  motivated by connections to log rank conjecture and the sensitivity conjecture, \cite{DineshS:18} explore connections between the sparsity, sensitivity and alteration complexity of a boolean function (which is essentially identical to the notion of $k$-monotonicity). \\

\paragraph{Organization} We proceed with the proofs of our lower bounds for the hypercube domain in Section \ref{sec:lower-bds} and with our algorithmic results for the hypergrid in Section \ref{sec:upper-bds}.

\section{Lower Bounds over the Hypercube}\label{sec:lower-bds}

\newcommand{\ra}{\rightarrow}
\newcommand{\etal}{{\em{et al}}}

 We prove all of our results in this section.  We first show  in Theorem \ref{thm:lb:k-vs-g(k,n):chain-tester}  that  the basic chain tester detects a violation  with  negligibly small probability, and hence  the $(k, g(k,n))$-problem is hard for chain testers.

\begin{theorem} \label{thm:lb:k-vs-g(k,n):chain-tester}
		Given $2\leq k \leq g(k,n) =\littleO{\sqrt n}$, there exist  $C > 0$, and a collection $\cl{F}$ of Boolean functions, such that
		\begin{enumerate}		
		\item[(i)] every $f \in \cl{F}$ is $\bigOmega{1}$-far from being $g(k,n)$-monotone, and 
		\item[(ii)] the probability that a uniformly random chain in $\{0,1\}^n$  detects a violation to $k$-monotonicity for $f$ is $\littleO{\exp\left(-C \frac{\sqrt n}{(g(k,n)/k)^2}\right)}$.
		\end{enumerate}
\end{theorem}

We then show that any other non-adaptive, one-sided tester gives rise to a chain tester, with only a small blowup in the query complexity.

\begin{theorem} \label{thm:lb:na:1s:k-vs-far-from-g(k,n)}		Any non-adaptive one-sided $q$-query $(k, g(k,n))$-tester for Boolean monotonicity over $\{0,1\}^n$  implies an $\bigO{q^{k+1}n}$-query tester that queries points on a distribution over chains, and 
succeeds with constant probability. In particular, if $p$ is the success probability of the basic chain tester, then $p=\Omega(1/q^{k+1})$.
\end{theorem}

Theorem \ref{thm:lb:k-vs-g(k,n):chain-tester} and Theorem \ref{thm:lb:na:1s:k-vs-far-from-g(k,n)} imply Theorem \ref{cor:lb:k-vs-g(k,n):general}.
 
\iffalse
\begin{corollary}\label{cor:lb:k-vs-g(k,n):general}
		Given $2\leq k \leq g(k,n) =\littleO{\sqrt n}$, testing non-adaptively with one-sided error whether a function is $k$-monotone or far from being $g(k,n)$-monotone requires $\exp\left({\Omega\left( \frac{\sqrt n}{(k+1)(g(k,n)/k)^2}\right)}\right)$ queries.
\end{corollary}
\fi

Instantiating $k$ and $g(k,n)$ in Theorem ~\ref{cor:lb:k-vs-g(k,n):general}, we obtain the following immediate corollaries.

\begin{corollary} \label{cor:lb:na:1s:2-vs-far-from-k}
		Any non-adaptive one-sided $(2,n^{0.01})$-tester for parametrized monotonicity requires  $\exp(\bigOmega{n^{0.48}})$ queries.
\end{corollary}

\begin{corollary} \label{cor:lb:na:1s:new}
		Let $2 \leq k = \littleO{\sqrt{n}}$.  Then any non-adaptive, one-sided tester for $k$-monotonicity  requires $\bigOmega{\exp\left(\frac{\sqrt n}{k+1}\right)}$ queries. 
\end{corollary}

%%%%%%%%%%%%%%%%%%%%%%%%%%%%%%%%%%%%%%%%%

Note that the lower bound of Corollary ~\ref{cor:lb:na:1s:new} is  $>\exp(\sqrt[4]{n})$ for $2 \leq k \leq \sqrt[4]{n}$. 
Using the previous lower bound of  $\bigOmega{\left( \frac{n}{k^2} \right)^{k/4}}$ for any $2 \leq k = \littleO{\sqrt n}$ from \cite{CGGKW:17}, we obtain the following immediate consequence.

\begin{corollary} \label{cor:lb:na:1s}
		Any non-adaptive, one sided tester for $k$-monotonicity, for $2 \leq k =\littleO{\sqrt n}$, requires $\Omega(\exp(\sqrt[4]{n}))$ queries.
\end{corollary}

\subsection{Proof of  Theorem \ref{thm:lb:k-vs-g(k,n):chain-tester}}

 Define the {\em weight} of an element $x$ in $\{0,1\}^n$, denoted $|x|$, to be the number of non-zero entries of $x$.

We first recall some standard useful facts.

\begin{fact} \label{fact:heavy-slice}
		The maximum value of $\binom{n}{t}$ occurs when $t = \lfloor n/2 \rfloor$, and this maximum value is less than $2 \cdot 2^n/\sqrt{n}$.

\end{fact}

\begin{fact} \label{fact:all-points-in-middle}
		There exists a constant $C>0$, such that for every $\eps > 0$, the number of points of  $\{0,1\}^n$ that with weight outside the \emph{middle  levels}  $[\frac{n}{2} - \sqrt n \log \frac{C}{\eps}, \frac{n}{2} + \sqrt n \log \frac{C}{\eps}]$  is at most $\eps 2^{n-1}$.
\end{fact}

In Section \ref{sec:hard:family}, we define our hard family, and show that every function in this family is indeed far from $g(k,n)$-monotonicity.  In Section \ref{sec:hard:family:chain-tester} we show that this family is  hard for the chain tester.

The hard family hides instances of a \emph{balanced blocks} function, which was previously used in~\cite{CGGKW:17} towards proving that testing $k$-monotonicity is at least as hard as testing monotonicity, even for an adaptive two-sided tester.

\newcommand{\BB}{\operatorname{BB}}
\begin{definition}[Balanced Blocks function]\label{def:balanced:block}
		For every $n$ and $\ell \leq o(\sqrt n)$, let us partition the vertex set of the Hamming cube into $\ell$ blocks $B_1, B_2, \ldots, B_{\ell}$ where every block $B_i$ consists of \emph{all} points in \emph{consecutive levels} of the Hamming cube, such that all of the blocks have \emph{roughly} the same size, i.e., for every $i \in \ell$, we have $$\left(1 - \frac{\ell}{\sqrt n} \right) \frac{2^n}{\ell} \leq \abs{B_i} \leq \left(1 + \frac{\ell}{\sqrt n} \right) \frac{2^n}{\ell}.\footnote{The blocks are stacked successively. The block with index i+1 lies immediately above the one with index i.}$$ 

		Then the corresponding \emph{balanced blocks function with $\ell$ blocks}, denoted\footnote{We will arbitrarily fix a function that satisfies these conditions.} $\operatorname{BB}(n, \ell) \colon \{0,1\}^n \to \{0,1\}$, is defined to be the blockwise constant function which takes value $1$ on all of $B_1$ and whose value alternates on consecutive blocks.\\

\end{definition} 
Note that $\BB(n, \ell)$ is a totally symmetric function: it is unchanged under permutations of its inputs.  Thus, we can partition $\{0,1,\ldots,n\}$ into $\ell$ intervals $I_1,I_2,\ldots,I_{\ell},$ such that $I_i$ is the set of Hamming levels that $B_i$ contains.

\cite{CGGKW:17} shows that Balanced Blocks functions satisfy a useful property that we soon recall  in Claim \ref{claim:bb:matching:size}, after making an important definition.

\begin{definition}[Violation hypergraph with respect to $\ell$-monotonicity (\cite{CGGKW:17})]
		Given a function $f\colon \{0,1\}^n \to \{0,1\}$, the \emph{violation hypergraph of $f$} is $H_{\rm viol}(f) = ( \{0,1\}^n, E(H_{\rm viol}) )$ where $(x_1, x_2, \cdots, x_{\ell+1}) \in E(H_{\rm viol})$ if the ordered $(\ell+1)$-tuple $x_1 \prec x_2 \prec \ldots \prec x_{\ell+1}$ (which is a $(\ell+1)$-uniform hyperedge) forms a violation to $\ell$-monotonicity in $f$.
		 A collection $M_h$ of pairwise disjoint $(\ell + 1)$-uniform hyperedges of the violation hypergraph  is said to form a \emph{violated matching}. 
\end{definition}

\begin{claim}[{\cite[Claim~3.8]{CGGKW:17}}]\label{claim:bb:matching:size}
		Let $h \eqdef \operatorname{BB}(n, \ell)$. Then $h$ is $(\ell-1)$-monotone and not $(\ell - 2)$-monotone. Furthermore, the violation graph of $h$ with respect to $(\ell - 2)$-monotonicity contains a violated matching of size at least $\frac{(1 - \littleO{1}) 2^n}{\ell}$, where every edge of the matching  $y_1 \preceq \dots \preceq y_{\ell-1}$ has $h(y_1) = 1$ and $h(y_i) \neq h(y_{i+1})$ for $1 \leq i < \ell-1$.\\
\end{claim}

This machinery allows us to deduce the following.

\begin{claim} \label{claim:bb:3k:blocks:far:k-mon}
		Let $h \eqdef \operatorname{BB}(n, 3k)$. Then $d(h, \mathcal{M}_k) \geq \bigOmega{1}$ \footnote{d(f,g) here denotes the Hamming distance between boolean functions $f$ and $g$ defined on the binary cube. Thus, \ref{claim:bb:3k:blocks:far:k-mon} says that the function $h$ is sufficiently far from being $k$-monotone.}.
\end{claim}

\begin{proofof} {Claim \ref{claim:bb:3k:blocks:far:k-mon}}

		By Claim \ref{claim:bb:matching:size},  $H_{\rm viol}(h)$ contains a matching $M_h$ of $(3k - 1)$-tuples of size $\frac{(1 - \littleO{1}) 2^n}{3k}$, and for every tuple in the matching $y_1 \preceq \dots \preceq y_{3k-1}$ we have $h(y_1) = 1$ and $h(y_i) \neq h(y_{i+1})$ for $1 \leq i < 3k-1$. We see that any $k$-monotone function close to $h$ must have at most $k$ flips within any such tuple, by definition. It follows that  any $k$-monotone function differs from $h$ in at least $k-1$ vertices in every tuple of $M_h$. Thus, the Hamming distance between $h$ and any $k$-monotone function is at least $$ \frac{(1 - \littleO{1}) 2^n}{3k} \cdot (k-1)\geq  \frac{ 2^n}{5}.$$

\end{proofof} 

We are now ready to describe the hard family.

\subsubsection{The Hard Family} \label{sec:hard:family}

In what follows, let $s \eqdef g(k,n)$. Also, let $r \eqdef \frac{g(k,n)}{k}$.
 We now describe the hard instance for $(k,s)$-testing.% using the chain tester.

Consider the partition of the set of indices  in $[n]$ into two different \emph{disjoint} sets, $L$ and $R$,  with sizes $|L| = n_L = n \cdot (1 - \frac{1}{625r^2})$ and $|R| = n_R = \frac{n}{625r^2}$, respectively.\footnote{For the sake of presentation, we ignore integrality issues where possible.} We will write $z \in \{0,1\}^n$ as $z = (x,y)$, where $x \in \{0,1\}^{|L|}$ and $y \in \{0,1\}^{|R|}$.  We define $\operatorname{MID}_L \eqdef \{i : \mid i - \frac{n_L}{2} \mid \leq \frac{\sqrt{n_L}}{100} \}$, to denote the set of ``balanced'' inputs restricted to the set of indices in $L$. 

Assuming $k$ is even\footnote{For odd $k$, the function is defined almost analogously -- the only difference is that $f_{n_L}(x,y) = 1 $ whenever $|x| > \frac{n_L}{2} + \frac{\sqrt n_L}{100}$. We make this  assumption throughout the paper.}, 
we define $  f_{n_L}:\{0,1\}^n\ra \{0,1\}$ by 

	  \[
			  f_{n_L}(x,y) \eqdef \begin{cases}
					  0, &  \text{ if } |x| \not\in \operatorname{MID}_L \\ 
					  \operatorname{BB}(n_R, 3s)(y), & \text{ otherwise i.e.,  } |x| \in \operatorname{MID}_L, \\            
			  \end{cases}
  \]

where $x \in \{0,1\}^{|L|}$ and $y \in \{0,1\}^{|R|}$.

 For $x \in \{0,1\}^L$, let us denote by $H_x$ the restriction of the hypercube $\{0,1\}^n$ to the points $(x, y)$, with $y \in \{0,1\}^{|R|}$. Note that for $x \in \operatorname{MID}_L$, the restriction of the function to $H_x$ is a copy of balanced blocks function on $n_R = n/625r^2$ variables with $3s$ blocks. 

\begin{claim} \label{claim:one-f-is-far-from-s-mon}
		The function $f=f_{n_L}$ is $\Omega(1)$-far from  being $s$-monotone.
\end{claim}

\begin{proof}

By Fact \ref{fact:all-points-in-middle}, picking  $\eps=\frac{1}{3}$, say, it follows that for a constant fraction of the points $x \in \{0,1\}^{n_L}$,  the function $f$ restricted to the cube $H_x$ is a  balanced blocks function on $3s$ blocks. 
 By Claim \ref{claim:bb:matching:size}, there is a matching of violations to $(3s-2)$-monotonicity within the violation graph on $H_x$, of size at least $\Omega(1) \cdot \frac{2^{n_R}}{3s}$. It follows that the there is a matching of violations to $(3s-2)$-monotonicity on the whole domain $\{0,1\}^n$ of size at least $\Omega(1) \frac{2^n}{3s}$. As in the proof of Claim \ref{claim:bb:3k:blocks:far:k-mon}, to produce a function that is $s$-monotone, one must change the value of $f$ in at least $s$ points of each matched hyper-edge. It follows that $f$ is $\Omega(1)$-far from being $s$-monotone.

\end{proof}

Our hard family of functions is the orbit of the function $ f_{n_L}$ under all the permutations of the variables.

\begin{definition} \label{def:hard:family}
		The family $\cl{F}_{k,s}$, parameterized by $k$ and $s$ is defined as follows.  Setting $n_L = (1 - \frac{k^2}{625s^2})n$, we define

		  $$\cl{F}_{k,s} \eqdef \setOfSuchThat{f_{n_L} \circ \pi_\sigma}{\sigma \in S_n}$$ 
		where $\pi_\sigma : \{0,1\}^n \to \{0,1\}^n$ is a permutation that sends the string $\{(a_1,a_2,\ldots,a_n)\}$ to $\{a_{\sigma(1)},a_{\sigma(2)},\ldots,a_\sigma(n)\}$ for a permutation $\sigma : [n] \to [n]$.	We omit the parameters $k$ and $s$ if it is clear from the context.\\
\end{definition}

We now observe that these functions are indeed far from being $s$-monotone. This follows since  $s$-monotonicity is closed under permutation of the variables, and by Claim \ref{claim:one-f-is-far-from-s-mon}.

\begin{claim} \label{claim:f-is-far-from-s-mon}
		Every $f \in \cl{F}_{k,s}$, $f$ is $\Omega(1)$-far from $s$-monotone.
\end{claim}

Therefore, we proved the distance property from Theorem \ref{thm:lb:k-vs-g(k,n):chain-tester}

\subsubsection{The Hard Family vs. the Chain Tester} \label{sec:hard:family:chain-tester}

Recalling that the basic chain tester picks a uniformly random chain in $\{0,1\}^n$, note that the distribution of the queries chosen by the chain tester is unchanged over permutations of the variables.  Thus, it suffices to analyze the probability that the chain tester discovers a violation to $k$-monotonicity on $f_{n_L}$.  We will show that this probability is very small if the quantity $s/k$ is small.

\begin{claim} \label{claim:k-vs-g(k,n):chain-tester}
		There exists a constant $C > 0$, such that the probability that a random chain reveals a violation to $k$-monotonicity in $f_{n_L}$ is at most $\exp\left(-C \frac{k^2}{s^2} \sqrt{n}\right)$.
\end{claim}

\newcommand{\MID}{\operatorname{MID}}
Let $Z$ be a fixed chain $0^n = z_0 \prec z_1 \prec z_2 \prec \cdots \prec z_n = 1^n$ in $\{0,1\}^n$.  Note that $f_{n_L}(z_i) = f_{n_L}(x_i, y_i) = 0$ if $|x_i| \notin \MID_L$.  Thus, if there is a violation to $k$-monotonicity in $Z$, then it can be found among points in $Z$ where $|x_i| \in \MID_L$.  Thus, a chain can only exhibit a violation on points $z_i = (x_i,y_i)$ where $n_L/2 - \sqrt{n_L}/100 \leq |x_i| \leq n_L/2 + \sqrt{n_L}/100$.  By definition, regardless of the exact structure of $x_i$ in this interval, $f_{n_L}(x_i,y_i) = \BB(n_R,3s)(y_i)$.  Since $\BB$ is a totally symmetric function, to determine if $Z$ exhibits a violation, it is enough to analyze the set

\[
		V(Z) \eqdef \{j : \text{there exists } z_i = (x_i,y_i) \in Z \text{ such that } |x_i| \in \MID_L \text{ and } |y_i| = j\}.
\]
We remark that for every chain $Z$, $V(Z)$ is a set of consecutive integers.

\begin{claim} \label{claim:violations-imply-large-v(z)}
		Suppose $Z$ contains a violation to $k$-monotonicity.  Then $|V(Z)| \geq k\sqrt{n_R}/(16s)$.
\end{claim}

\begin{proof}
		By Fact ~\ref{fact:heavy-slice}, the maximum value of $Pr_{\by \sim \{0,1\}^{n_R}}[|\by| = t]$ over values of $t$ is $2/\sqrt{n_R}$.  Since $BB(n_R,3s)$ has $3s$ blocks, the number of consecutive levels of $I_i$ in any block $B_i$ must satisfy

		\[
				(2/\sqrt{n_R})|I_i| \geq \frac{1}{3s} (1 - o(1)) \geq \frac{1}{4s},
		\]
		so $|I_i| \geq \sqrt{n_R}/(8s)$.  To see a violation to $k$-monotonicity, the chain $Z$ must contain points from each Hamming level in $k-1$ complete blocks, so this requires $|V(Z)| \geq (k-1)\sqrt{n_R}/(8s) \geq k\sqrt{n_R}/(16s),$ as claimed.

\end{proof}	

We will show that that $|V(\bZ)|$ reaching this value is very unlikely for a random chain $\bZ$.
Let $\bZ$ be a random chain $0^n \prec \bz_1 \prec \cdots \prec \bz_{n-1} \prec 1^n$. 

\begin{proofof} {Claim \ref{claim:k-vs-g(k,n):chain-tester}}
	The proof follows  from Claim \ref{claim:violations-imply-large-v(z)} and the 
	 following claim.

		\begin{claim} \label{claim:random:chain:low:prob:violation}
				Let $\bZ$ be a random chain.  Then $\Pr[|V(\bZ)| \geq k\sqrt{n_R}/(20s)] \leq \exp\left(-\dfrac{0.00009}{r^2}\sqrt{n}\right)$.
		
		\end{claim}

		\begin{proof}
				Let $j$ be the smallest index such that $\bz_j = (\bx_j,\by_j)$ satisfies $|\bx_j| = n_L/2 - \sqrt{n_L}/100$.  This is the index where the chain enters the region where it could find violations.
				
				Let $w$ be the largest index such that $\bz_w = (\bx_w,\by_w)$ satisfies $|\bx_w| = n_L/2 + \sqrt{n_L}/100$.  If $\bZ$ contains a violation to $k$-monotonicity, then it must occur on the subchain $\bz_{j-1} \prec \bz_j \prec \cdots \prec \bz_w \prec \bz_{w+1}$.  By construction, we have $f(\bz_\ell) = 0$ if $\ell \leq j-1$ or $\ell \geq w+1$.  Further, $V(\bZ) = \{|\by_j|, |\by_{j+1}|, \ldots, |\by_{w-1}|, |\by_w|\}$, and $|V(\bZ)| = |\by_w| - |\by_j| + 1$.  Thus, to prove the claim, it satisfies to analyze $|\by_w| - |\by_j|$.  Note $w-j$ is exactly $\sqrt{n_L}/50 + |V(\bZ)| - 1$; this accounts for $\sqrt{n_L}/50$ flips of variables in $L$ and $|V(\bZ)| - 1$ flips of variables in $R$.  Informally, we want to show that the ratio of the number of variables flipped in $L$ to number of variables flipped in $R$ is, with very high probability, too small for the chain tester to succeed in finding a violation to $k$-monotonicity.
				
				To simplify our analysis, we will not work directly with $w$.  Instead, define $j$ as above, but consider $\bz_{j'}=(\bx_{j'}, \by_{j'}),$ where $j' = j + \sqrt{n}/3$.  We will show that, with high probability, $j' \geq w$, and $|\by_{j'}| - |\by_j|$ (and thus $|\by_w| - |\by_j|$) is small. 
				
				We claim that the value of $|\by_{j'}| - |\by_j|$ is a random variable with a (random) hypergeometric distribution.  Indeed, to draw a random variable distributed as this difference, we construct the following experiment that simulates the behavior of a random chain with respect to the function $f$:

				\begin{itemize}
						\item The chain tester picks $\sqrt{n}/3$ coordinates from the $n-|\bz_j|$ coordinates set to $0$.  
						\item $n_R-|\by_j|$ of these coordinates are ``successes'' for the chain tester, which correspond to flipping variables in $R$, and
						\item $n_L-|\bx_j|$ of these coordinates are ``failures'' for the chain tester, which correspond to flipping variables in $L$. 

				\end{itemize}

			Let $H(u,N,t,i)$ denote the probability of seeing exactly $i$ successes in $t$ independent samples, drawn uniformly and without replacement, from a population of $N$ objects containing exactly $u$ successes.
		
				The chain tester is most likely to see successes in the above experiment if $|\by_j| = 0$; we will assume that this happens. As seen in the proof of Claim ~\ref{claim:violations-imply-large-v(z)},  in order  to successfully reject $f$, the chain  must witness at least $\frac{k \sqrt {n_R}}{16s}$ successes. 	Let $t = \frac{k\sqrt{n_R}}{20s}$.
			
			Note that if the number of successes is $i < t$, then the number of failures is at least $\frac{\sqrt n}{3} - t \geq \frac{\sqrt n}{3} - \frac{\sqrt n}{500} > \frac{\sqrt{n_L}}{50}$, and so in this case we have $|V(\bZ)| < t$; this corresponds to the chain missing a complete balanced block. It follows that the proof reduces to upper bounding the quantity 
			$$\Pr[|V(\bZ)| \geq t] = \displaystyle\sum_{i \geq t} H(n_R,n_L/2+\sqrt{n_L}/100+n_R,\sqrt{n}/3,i).$$	
		
		We analyze the above quantity using a Chernoff bound for hypergeometric random variables, where $\bX=|\by_{j'}| - |\by_j|$.

				\begin{theorem}[Theorem~1.17 in~\cite{Auger2011}]
						Let $\bX$ be a hypergeometrically distributed random variable.  Then% for $0 < \delta < 1$, we have
						\[
								\Pr[\bX \geq \frac54 \cdot \expect{\bX}] \leq \exp(-\expect{\bX}/48).
						\]
				\end{theorem}
				\newcommand{\E}{\expect}
				
				We use the following claim.
				
				\begin{claim}
				$\frac4{15} \cdot t \leq \E{\bX}\leq  \frac45 \cdot t$.
				
				\end{claim}
				
				\begin{proof}
				
				Standard facts about the hypergeometric distribution imply that
						$$\E{\bX}=\frac{\sqrt n}{3} \cdot \dfrac{n_R}{n_L/2+\sqrt{n_L}/100+n_R}.$$

				Recall that $r=s/k\geq 1$, $n_L=n(1-1/(625 r^2)) > 2n/3$, $n_R=n/(625 r^2)$, and  $t = \frac{k\sqrt n_R}{20s}=\frac{\sqrt{n_R}}{20r} = \frac{\sqrt{n}}{500r^2}$.
				It follows that $$n_L/2+\sqrt{n_L}/100+n_R > n_L/2 > n/3.$$
				Therefore
				$$\E{\bX}<\frac{\sqrt n}{3} \cdot \frac{3 n_R}{n}=\sqrt{n}\cdot \frac{1}{625r^2}=\frac{4}{5}\frac{\sqrt{n}}{500r^2} = \frac{4}{5} \cdot t.$$
				Since $n_L/2 + \sqrt{n_L}/100 + n_R < n$,
				 
				$$\E{\bX}>\frac{\sqrt n}{3} \cdot \frac{n_R}{n}=\frac{\sqrt{n}}{3} \cdot \frac{1}{625r^2}=\frac{4}{15}\frac{\sqrt{n}}{500r^2} = \frac{4}{15}\cdot t.$$
				
				\end{proof}
			We can now conclude that  			
				\begin{eqnarray*}
				\Pr[\bX\geq t]&=&\Pr[|V(\bZ)| \geq t] = \displaystyle\sum_{i \geq t} H(n_R,n_L/2+\sqrt{n_L}/100,\sqrt{n}/3,i)\\
				                  &=& \Pr\left[X\geq \E{\bX}\cdot \frac{t}{\E{\bX}}\right]\\
				                  &\leq& \Pr\left[X\geq \E{\bX}\cdot \frac{5}{4}\right]\\
				                     &\leq& \exp(-\E{\bX}/48)\\
				                     &\leq& \exp(-t/180)=\exp\left(\frac{\sqrt{n}}{90000r^2}\right).
				\end{eqnarray*}

		\end{proof}

\end{proofof}

\subsection{The reduction}

We now prove \autoref{thm:lb:na:1s:k-vs-far-from-g(k,n)}. As mentioned,   \autoref{thm:lb:na:1s:k-vs-far-from-g(k,n)} and  \autoref{thm:lb:k-vs-g(k,n):chain-tester} immediately imply \autoref{cor:lb:k-vs-g(k,n):general}.

\begin{proof}[Proof  of \autoref{thm:lb:na:1s:k-vs-far-from-g(k,n)}]
We show that given a $q$-query non-adaptive, one-sided $(k, s)$-tester, one can obtain a $\bigO{q^{k+1} n}$-query  $(k, s)$-tester that only queries values on a distribution over  random chains.

Let $T$ be a $q$-query non-adaptive, one-sided $(k,s)$-monotonicity tester.
Therefore, $T$ accepts functions that are $k$-monotone, and rejects functions that are $\eps$-far from being $s$-monotone with probability $2/3$.

Define a tester $T'$ that on input a function $f$ does the following:  it first gets the queries of $T$,  then for each $(k+1)$-tuple  $q_1 \prec q_2 \prec \cdots \prec q_{k+1}$, $T'$ queries an entire uniformly random chain from $0^n$ to $1^n$, conditioned on containing these $k+1$ points.  Therefore, $T'$ is also one-sided, makes $O({q \choose{k+1}} n) = O(q^{k+1}n)$ queries, and its success probability is no less than \footnote{We assume that every query made by $T$ belongs to at least one $(k+1)$-tuple.  Queries that do not are of no help to $T$, since these queries can not be part of a violation to $k$-monotonicity discovered by $T$, and we are assuming that $T$ is non-adaptive and one-sided.} the success probability of $T$.

\newcommand{\bsigma}{\boldsymbol{\sigma}}

Define $T''$ that on input $f$ picks a random permutation $\bsigma\colon [n] \to [n]$ and then applies the queries of $T'$  to the function $f\circ \pi_{\bsigma}$ (where $\pi_{\bsigma}$ is defined as in Definition \ref{def:hard:family}). This means that if $T'$ queries $q$, $T''$ queries $\pi_{\bsigma}(q)$. Then $T''$ ignores what $T'$ answers and only rejects if it finds a violation on any one of the chains.  

Note that if $f$ is $k$-monotone, then so is $f\circ \pi_{\bsigma}$, and if $f$ is $\eps$-far from being $s$-monotone, then so is $f\circ \pi_{\bsigma}$. 

Therefore, $T''$ is one-sided, non-adaptive, and makes $O(q^{k+1}n)$ queries. Since $T'$ is one-sided, it can only reject if it finds a $(k+1)$-tuple forming a  violation to $k$-monotonicity. So if $T'$ rejects, so does $T''$, and it follows that  the success probability of $T''$ is at least the success probability of $T'$, which is at least $2/3$.

We now claim that the queries of $T''$ are distributed as $O({q \choose{k+1}})$ uniformly random chains.  While the marginal distribution for each individual chain is the uniform distribution over chains, the joint distribution over these chains is not necessarily independent. 
Suppose $T'$ queries points $q_1, q_2, \cdots q_{k+1}$ with $q_1 \prec q_2 \prec \cdots \prec q_{k+1}$.
Then $\pi_{\bsigma}(q_i)$ is a uniformly random point on  its weight level, and   $\pi_{\bsigma}(q_1)\prec \pi_{\bsigma}(q_2) \prec \ldots \prec \pi_{\bsigma}(q_{k+1})$. It follows that a chain chosen uniformly at random conditioned on passing through these points is a uniformly random chain in $\{0,1\}^n$.

Let $p$ be the success probability of the basic chain tester that picks a uniformly random chain in $\{0,1\}^n$ and rejects only if it finds a violation to $k$-monotonicity.  Taking a union bound over the chains chosen by $T''$, the success probability of $T''$ is at most $p\cdot{q \choose{k+1}} \leq p \cdot q^{k+1}.$
It follows that $p\cdot q^{k+1}\geq 2/3$, from which it easily follows that $p = \Omega(q^{1/(k+1)})$, concluding the proof.
\end{proof}

\section{Upper Bounds over the Hypergrid}\label{sec:upper-bds}

In this section we  prove Theorem \ref{thm:ub}.
\iffalse
\begin{theorem}
	There is a non-adaptive two-sided tester using $\poly(1/\eps)$ samples that accepts all $k$-monotone functions and rejects all functions that are $\eps$-far from $2kd^2/\eps$-monotone.	
\end{theorem}
\fi
\newcommand{\calB}{\mathcal{B}}
\renewcommand{\dist}[2]{d(#1,#2)}
In this section, we consider Boolean functions over the hypergrid $[n]^d$.  For convenience, we will define $[n] = \{0,1,2,\ldots,n-1\}$. 
Assuming $m$ divides $n$, we define $\calB_{m,n} : [n]^d \to [m]^d$ be the mapping such that $\calB_{m,n}(y)_i = \lfloor y_i/m \rfloor$ for $1 \leq i \leq d$.  For $x \in [m]^d$, we define $\calB^{-1}_{m,n}(x)$ to be the inverse image of $x$ under $\calB_{m,n}$.  Specifically, $\calB^{-1}_{m,n}(x)$ is the set of points of the form $m \cdot x + [n/m]^d$, using the standard definitions of scalar multiplication and coordinatewise addition.  That is, $\calB^{-1}_{m,n}$ is a ``coset'' of $[n/m]^d$ points in $[n]^d$.  We will call these cosets \emph{blocks}, and we will say that $h : [n]^d \to \{0,1\}$ is an $m$-block function if it is constant on $\calB^{-1}_{m,n}(x)$ for every $x \in [m]^d$.  For readability, we will often suppress the dependence on $m$ and $n$.

\begin{claim}[Claim 7.1, \cite{CGGKW:17}]\label{claim:mono-coarsening}
	Suppose $f : [n]^d \to \{0,1\}$ is $k$-monotone.  Then there is an $m$-block function $h : [n]^d \to \{0,1\}$ such that $\dist{f}{h} < kd/m$.
\end{claim}

\begin{claim}\label{claim:block-to-mono}
	Suppose $h : [n]^d \to \{0,1\}$ is an $m$-block function.  Then $h$ is $((m-1)d-1)$-monotone.
\end{claim}

\begin{proof}
	Without loss of generality, we assume that $h(0^d) = 0$.
	Suppose for the sake of contradiction that $h$ is an $m$-block function such that $h$ contains a violation to $((m-1)d-1)$-monotonicity.  Equivalently, there exists $y_0 \prec y_1 \prec \cdots \prec y_{(m-1)d-1}$ in $[n]^d$ such that $h(y_0) = 1$ and $h(y_i) \neq h(y_{i+1})$ for $0 \leq i \leq (m-1)d-2$.  Since $h$ is constant on each block, we have no two $y_i$'s are in the same block.  Thus the set $\{\calB_{m,n}(y_i) : 0 \leq i \leq (m-1)d-1 \}$ contains $(m-1)d$ distinct vectors in $[m]^d$ that can be totally ordered.  This implies that $0^d$ and $(m-1)^d$ (this is a vector of $d$ coordinates, all of which are $m-1$) are in this set.  Clearly, $0^d$ is the ``smallest'' vector in this total order, and it follows from our definition of violation that $h(0^d) = 1$.  This is a contradiction, since we assumed at the outset that $h(0^d) = 0$.  Thus $h$ does not contain a violation to $((m-1)d-1)$-monotonicity, and $h$ is $((m-1)d-1)$-monotone.
\end{proof}

We define the $m$-block-coarsening of a function $f : [n]^d \to \{0,1\}$ to be the $m$-block function $h : [n]^d \to \{0,1\}$ such that $\dist{f}{h}$ is as small as possible.  We would like to use query access to $f$ to get query access to $h$, but it is not guaranteed we can do this.  Rather, for each $x \in [m]^d$, we randomly select a set $\bS_x$ of points in the block $\calB^{-1}(x)$, where we choose $|\bS_x|$ to be large enough such that

\[
\left\lvert \Pr_{\bz \sim \bS_{\calB(y)}}[f(\bz) = 0] - \Pr_{\bz \sim \calB^{-1}(\calB(y))}[f(\bz) = 0] \right\rvert \leq 1/9,
\]
with high probability.  Our target query complexity is independent of $m$, so we can \emph{not} necessarily query points from each $\bS_x$; these points merely allow us to talk about a specific (randomly chosen) function.  We denote by $\bh' : [n]^d \to \{0,1\}$ the $m$-block function such that

\[
\bh'(y) = \mathop{\arg\!\max}_{b \in \{0,1\}} \Pr_{\bz \sim \bS_{\calB(y)}}[f(\bz) = b],
\]
breaking ties arbitrarily. In other words, $\bh'$ can be thought of as the $m$-block function obtained by (randomized) ``greedy'' coarsening of $f$ where we greedily select the majority bit (or the winning bit) over the random sample $\bS_x$ to be the value $h'$ over the corresponding block.

\begin{claim}
We have $\dist{f}{\bh'} \leq \frac54\dist{f}{h}$.
\end{claim}

\begin{proof}
Let $B$ be a block such that $h$ and $\bh'$ disagree.  Then the wrong bit was estimated to be the majority value of $f$ when restricted to $B$.  By our construction of $\bh'$, we must have

\[
\left\lvert \Pr_{\by \sim B}[f(\by) = h(\by)] - \Pr_{\by \sim B}[f(\by) = \bh'(\by)] \right\rvert \leq 1/9,
\]
since exactly one of $h(\by)$ and $\bh'(\by)$ is $0$ and the other is $1$.  It follows that $\Pr_{\by \sim B}[f(\by) \neq h(\by)] \geq 4/9$ and
$\Pr_{\by \sim B}[f(\by) \neq \bh'(\by)] \leq 5/9$.  Combining these inequalities, we get $\Pr_{\by \sim B}[f(\by) \neq \bh'(\by)] \leq \frac54 \Pr_{\by \sim B}[f(\by) \neq h(\by)]$.

Clearly, if $h$ and $\bh'$ do not disagree on $B$, then the previous probabilities are equal and the inequality holds. 
%Since $\dist{f}{h} = E[\Pr_{\by \sim B}[f(\by) \neq h(\by)]] \geq 4/9$
Thus, for all blocks $B$, $\Pr_{\by \sim B}[f(\by) \neq \bh'(\by)] \leq \frac54 \Pr_{\by \sim B}[f(\by) \neq h(\by)]$.  The claim follows by taking the expected value of each side over a uniformly chosen block.

%\[
%\dist{h}{\bh'} = E[\Pr_{\by \sim B}[f(\by) = h(\by)] - \Pr_{\by \sim B}[f(\by) = \bh'(\by)]] \leq E[\eps/10] = \eps/10.
%\]

\end{proof}

\begin{proof}[Proof of Theorem~\ref{thm:ub}] In our tester, we set $m = (2kd^2/\eps + 1)/d + 1$, and the tester simply estimates $\dist{f}{\bh'}$ to within $\pm \eps/8$, which can be done with $\widetilde{O}(1/\eps^2)$ queries. Note that query access to $\bh'$ can be simulated because $\bh'$ is the result of the randomized greedy coarsening described earlier. 
By Claim~\ref{claim:mono-coarsening}, if $f$ is $k$-monotone, then there is an $m$-block function $h$ such that

\[
\dist{f}{h} < kd/m = kd/((2kd^2/\eps + 1)/d + 1) < kd^2/(2kd^2/\eps) = \eps/2,
\]
and it follows that $\dist{f}{\bh'} \leq \frac54\dist{f}{h} \leq \frac54(\eps/2) = 5\eps/8$.
By Claim~\ref{claim:block-to-mono}, if $f$ is far from $2kd^2/\eps$-monotone, then it is $\eps$-far from every $m$-block function, as for our setting of $m$, we have $((m-1)d-1) = 2kd^2/\eps$.
It follows that $\dist{f}{\bh'} \geq \eps$.  Thus, the tester correctly accepts if the estimate of $\dist{f}{\bh'}$ is at most $3\eps/4$ and correctly rejects if this estimate is at least $7\eps/8$.

\end{proof}

\paragraph{Acknowledgments.} We thank  our collaborators Cl\'ement Canonne and Siyao Guo, who have gracefully refused to co-author this paper.
%\nocite{*}
\bibliographystyle{alpha}
\bibliography{references} 

\end{document}